\documentclass[conference,a4paper]{IEEEtran}

\usepackage{amsmath}
\usepackage{amssymb}
\usepackage{amsfonts}
\usepackage{amsthm}
\usepackage{graphicx}

\newcommand{\R}{\mathbb{R}}
\newcommand{\N}{\mathbb{N}}

\newcommand{\one}{\mathbf{1}}

\newcommand{\multisetds}[2]{\bigg(\kern-.4em\binom{#1}{#2}\kern-.4em\bigg)}
\newcommand{\multisetin}[2]{\big(\kern-.3em\binom{#1}{#2}\kern-.3em\big)}
\newcommand{\multisetix}[2]{\left(\kern-.2em\binom{#1}{#2}\kern-.2em\right)}

\newcommand{\PR}[1]{P\big[#1\big]}

\newcommand{\mats}{\ensuremath{\mathcal{S}}}

\ifx\eqref\undefined
    \newcommand{\eqref}[1]{~(\ref{#1})}
\fi
\ifx\mod\undefined
    \def\mod{\mathop{\rm mod}}
\fi

\def\exp{\mathop{\rm exp}}

\def\tr{\mathop{\rm tr}}

\def\eqdef{\stackrel{\triangle}{=}}

\def\pgeq{\succcurlyeq}

\ifx\lesssim\undefined
\def\simleq{{{\mskip 3mu plus 2mu minus 1mu%
    \setbox0=\hbox{$\mathchar"013C$}%
    \raise.2ex\copy0\kern-\wd0%
    \lower0.9ex\hbox{$\mathchar"0218$}}\mskip 3mu plus 2mu minus 1mu}}
\else
\def\simleq{\lesssim}
\fi
\ifx\gtrsim\undefined
\def\simgeq{{{\mskip 3mu plus 2mu minus 1mu%
    \setbox0=\hbox{$\mathchar"013E$}%
    \raise.2ex\copy0\kern-\wd0%
    \lower0.9ex\hbox{$\mathchar"0218$}}\mskip 3mu plus 2mu minus 1mu}}
\else
\def\simgeq{\gtrsim}
\fi

\DeclareMathOperator{\img}{Im}
\newtheorem{theorem}{Theorem}
\newtheorem{remark}{Remark}
\newtheorem{proposition}{Proposition}

\newtheorem{corollary}{Corollary}

\newcommand{\bcomment}[1]{}
\newcommand{\TODO}[1]{}

\begin{document}
\title{Rate-distance tradeoff for codes above graph capacity}

\IEEEoverridecommandlockouts

\author{\IEEEauthorblockN{Daniel Cullina}
\IEEEauthorblockA{
University of Illinois at Urbana-Champaign\\
cullina@illinois.edu}
\and
\IEEEauthorblockN{Marco Dalai}
\IEEEauthorblockA{University of Brescia\\
marco.dalai@unibs.it}
\and
\IEEEauthorblockN{Yury Polyanskiy}
\IEEEauthorblockA{Massachusetts Institute of Technology\\
yp@mit.edu}%
\thanks{
The research was supported by the NSF grant CCF-13-18620 and NSF Center for Science of Information (CSoI) 
under grant agreement CCF-09-39370. The work was partially done while visiting the Simons Institute for the
    Theory of Computing at UC Berkeley, whose support is gratefully acknowledged. }
}

\maketitle

\begin{abstract}
The capacity of a graph is defined as the rate of exponential growth of independent sets in the strong powers of the graph.
In the strong power an edge connects two sequences if at each position their letters are equal or adjacent.
We consider a variation of the problem where edges in the power graphs are removed between sequences which differ in more than a fraction $\delta$ of coordinates.
The proposed generalization can be interpreted as the problem of determining the highest rate of zero undetected-error communication over a link with adversarial noise, where only a fraction $\delta$ of symbols can be perturbed and only some substitutions are allowed.

We derive lower bounds on achievable rates by combining graph homomorphisms with a graph-theoretic generalization of the Gilbert-Varshamov bound.
We then give an upper bound, based on Delsarte's linear programming approach, which combines Lov\'asz' theta function with the construction used by McEliece et al. for bounding the minimum distance of codes in Hamming spaces.
\end{abstract}


\section{Introduction}
The problem we consider is the following. Given a graph $G$ we define a semimetric on the vertex set $V(G)$ 
\[ d(v,v') = \begin{cases} 0, & v=v',\\
            1, & \{v, v'\} \in E(G) ,\\
            \infty, & \mbox{otherwise.} \end{cases}\]
We extend this semimetric additively to the Cartesian products $V(G)^n$ and define a graph $G(n,d)$ as follows
\begin{IEEEeqnarray*}{rCl}
V(G(n,d)) &=& V(G)^n\\
E(G(n,d)) &=& \left\{(x,x'): d(x,x') \eqdef \sum_{j=1}^n d(x_j, x_j') \le d \right\}\,.
\end{IEEEeqnarray*}
The goal is to determine (bounds on)
\[ R^*(G, \delta) \eqdef \limsup_{n\to\infty} \frac{1}{n} \log \alpha(G(n,\delta n))\,.\]

Note that $G(n,d)$ corresponds to the graph obtained by removing in the strong power graph $G^n$ edges between sequences which differ in more than $d$ positions. On one hand, this problem is a specialization of the general one considered in \cite{dalai-TIT-2015}. On
the other hand, it is a natural generalization of the two classically studied ones:
\begin{enumerate}[itemindent=0.5cm, leftmargin=0cm]
\item {\em Shannon capacity of a graph}~\cite{shannon-1956}, which corresponds to $\delta=1$. The best general upper bound is~\cite{lovasz-1979}
    \begin{equation}\label{eq:lovasz}
     R^*(G, 1) \le \log \theta_L(G)\,,
\end{equation}    
where $\theta_L$ is the Lovasz theta function.
\item {\em Rate-Distance tradeoff in Hamming spaces}, which corresponds to $G=K_q$ (the clique). Here the two bounds we mention are 
\begin{equation}\label{eq:hamming}
    R_{GV}(q,\delta) \le R^*(K_q,\delta) \le R_{LP1}(q,\delta)\,,
\end{equation}
where for $\delta < 1-\frac{1}{q}$
\begin{align} R_{GV}(q,\delta) &\eqdef \log q - H_q(\delta)\,,\\
   R_{LP1}(q,\delta) &\eqdef H_q\left(\frac{(q-1)-(q-2)\delta - 2\sqrt{(q-1)\delta(1-\delta)}}{q}\right)\,,\label{eq:rlp1}\\
   H_q(x) &\eqdef x \log(q-1) - x \log x - (1-x)\log(1-x)\,.
\end{align}   
For $\delta\ge 1-\frac{1}{q}$ both $R_{GV}$ and $R_{LP1}$ equal zero.%
\footnote{Better bounds also exist: an improved upper bound for for small $\delta$'s was found by Aaltonen~\cite{aaltonen1990new}, and an
improved lower bound for large $q$'s and some range of $\delta$'s is shown via algebraic-geometric codes~\cite{tsfasman1982modular}.}
We refer the point $\delta = 1-\frac{1}{q}$ as the Plotkin point.%
\footnote{The Plotkin bound is the simplest upper bound that establishes that $R^*(K_q,\delta) = 0$ for $\delta \ge 1-\frac{1}{q}$.}
\end{enumerate}
The proposed problem can be interpreted as the natural extension of the notion of rate-distance tradeoff to the case where only some substitutions are allowed.

In this paper we derive both upper and lower bound on $R^*(G,\delta)$ for different classes of graphs. In particular, among other more specific bounds, we prove that if $G$ is vertex-transitive with independence number $\alpha(G)$, then
\begin{equation}
R^*(G, \delta) \geq \log \alpha(G) + R_{GV}\left(\frac{|V(G)|}{\alpha(G)}, \delta\right)\label{eq:LB-general}
\end{equation}
and if $G$ is also edge-transitive, then
\begin{equation}
R^*(G, \delta) \leq \log \theta_L(G) + R_{LP1}\left(\frac{|V(G)|}{\theta_L(G)}, \delta\right)\,.\label{eq:UB-general}
\end{equation}
A graph is vertex-transitive if its automorphism group is transitive on the vertex set and edge-transitive if the group is transitive on the edge set.
These two bounds can be interpreted as simultaneous generalizations of equation \eqref{eq:hamming}, since $\alpha(K_q)=\theta_L(K_q)=1$, and of the known bounds on the graph capacity
\begin{equation}
\log \alpha(G)\leq R^*(G,1)\leq \log \theta_L.
\end{equation}
Note however that for general symmetric graphs, the quantities   $|V(G)|/\alpha(G)$ and $|V(G)|/\theta_L(G)$ which appear in the usual role of alphabet size, are in general not integers. To the best of our knowledge this is the first appearance of non-integer quantities in this role.

The main tools used for our achievability results are graph homomorphisms and a graph-theoretic generalization of the Gilbert-Varshamov bound.
Our main converse, instead, is obtained by adapting ideas from the Delsarte's linear programming bound.

We use the following standard graph theoretic notation \cite{west_introduction_2001}. 
For a graph $G$, we denote by $\alpha(G)$ the size of a largest independent set and by $\omega(G)$ the size of a largest clique.  We denote with $\chi(G)$ and $\chi^*(G)$ the chromatic and the fractional chromatic number respectively. Finally, $\theta^*(G)=\chi^*(\overline{G})$ is the fractional clique covering number.

\section{Preview Examples}
\label{sec:examples}
\begin{figure}
\includegraphics[width=\columnwidth]{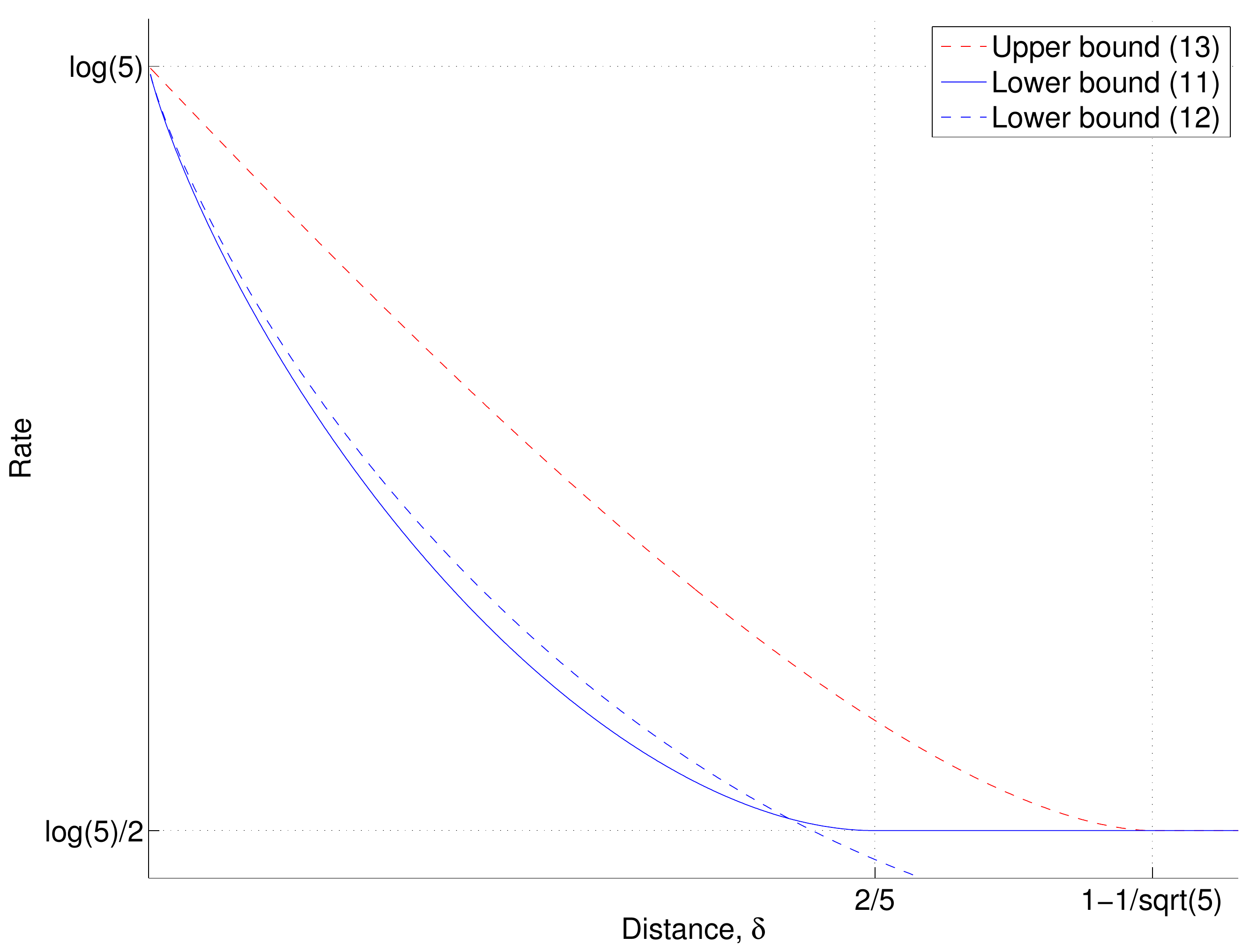}
\caption{Bounds on $R^*(G,d)$ for the pentagon.}
\label{fig:pentagon}
\end{figure}

Let $K_q$ be the complete graph on $q$ vertices and let $G$ be the disjoint union of $m$ copies of $K_q$, i.e. $G = K_q+K_q+\cdots+K_q = mK_q$. 
Then according to Proposition \ref{prop:tight-cover} below we have
\begin{equation}
R^*(G,\delta)=\log(m)+R^*(K_q,\delta).
\end{equation}
Here the situation is quite intuitive. $G^n$ is itself a disjoint union of $m^n$ equally sized cliques, each being equivalent, under the considered semimetric, to a $q$-ary Hamming space. Thus, $\log(m^n)$ bits are associated to the choice of the clique and within each clique we can additionally pack sequences at distance $n\delta$ at an asymptotic rate $ R^*(K_q,\delta)$.

Consider now the case of an even cycle $C_{2m}$, which might be interpreted as a first example of non-disjoint cliques of size $2$.
Proposition \ref{prop:tight-cover} says that the problem still reduces to the binary case:
\begin{equation}
R^*(C_{2m},\delta)=\log(m)+R^*(K_2,\delta).
\end{equation}
Here as well we might think in some sense of having partitioned our global space in $m^n$ binary Hamming spaces, though a more careful analysis is required to appreciate the details.

For odd cycles the situation is different. The best we can prove is based on equations \eqref{eq:LB-general} and
\eqref{eq:UB-general}, where \eqref{eq:LB-general} can in general be improved by also considering powers of $G$ (see
Proposition \ref{prop:using-powers} below). For the pentagon, for example, equation \eqref{eq:LB-general} applied to $C_5$ and to $C_5^2$ leads respectively to
\begin{align}
R^*(C_5,\delta) & \geq \log(2)+R_{GV}(5/2,\delta)\\
R^*(C_5,\delta)& \geq \frac{1}{2}\log(5)+\frac{1}{2}R_{GV}(5,2\delta),\label{eq:LB-C_5^2}
\end{align}
the first being stronger for $\delta\lesssim 0.353$. Equation \eqref{eq:UB-general} gives
\begin{equation}
R^*(C_5,\delta)\leq \frac{1}{2}\log(5)+R_{LP1}(\sqrt{5},\delta).
\end{equation}
Figure \ref{fig:pentagon} shows the corresponding plots. 

In this case, an interpretation of the bounds \eqref{eq:LB-general} and \eqref{eq:UB-general} in terms of a partition of the space into some number of Hamming-like spaces requires fractional values for their alphabet sizes.
Note that in the case of $C_5$, our bounds do not pin down what might be in our context the equivalent of the Plotkin point, i.e. the value $\delta_P$ such that $R^*(G,\delta) = R^*(G,1)$ for $\delta \geq \delta_P$ and $R^*(G,\delta) > R^*(G,1)$ for $\delta < \delta_P$.

The gap between bounds observed for $C_5$ might not be surprising, since odd cycles are notoriously hard to deal with in general.
Another very simple example gives a feeling of the subtleties which one should expect in this context. 
Consider the simplest possible case of disjoint union of unequally sized cliques: $G=K_1+K_2$. Proposition \ref{sec:unequal-cliques} gives
\begin{equation}
R^*(K_1+K_2,\delta)=\max_{0 \leq \lambda \leq 1} \left[ H_2(\lambda)+\lambda R^*(K_2,\delta/\lambda) \right].\label{eq:K2K1}
\end{equation}
Bounding $R^*(K_2,\cdot)$ via~\eqref{eq:hamming} we infer that the lower bound achieves $R=\log(2)$ at $\delta=1/4$ (obtained
for $\lambda=1/2$), while the upper bound only says 
$R\leq \log(2)$ at $d\gtrsim 0.2568$. Thus, determining the Plotkin point
even for this simple graph is as hard as improving the best known bound on $R^*(K_2,\delta)$! 

In view of the hardness of the case of $K_1+K_2$ it may be surprising that we can instead establish the Plotkin point for the
much more complicated Kneser graphs. Let $K_{c,a}$ be the graph whose vertices are the subsets of $\{1,2,\ldots,c\}$ of size $a$, two vertices being adjacent if and only if they are disjoint.
For these graphs we have $\alpha(K_{c,a}) = \binom{c-1}{a-1} = k$ and 
\begin{equation*}
     \log k + R_{GV} \left( \frac{c}{a}, \delta \right) 
 \le R^*(K_{c,a}, \delta) 
 \le \log k + R_{LP1} \left( \frac{c}{a}, \delta \right)\,. 
\end{equation*}
and the Plotkin point is at $\delta=1-\frac{a}{c}$.
Similar conclusions hold for any other edge-transitive graph $G$ with $\alpha(G)=\theta_L(G)$.

\section{Achievability bounds}
\label{sec:achievability}
Let $N(u)$ be the neighborhood of a vertex $u$, $N[u] = N(u) \cup \{u\}$, $N(S) = \cup_{u \in S} N(u)$, and $N[S] = N(S)
\cup S$. The following is a generalization of the standard Gilbert-Varshamov bound:

\begin{proposition}
\label{prop:general-lb}
Let $\mats$ be a family of independent sets in $G$ and let $S$ be a random variable with state space $\mats$.
Then $\alpha(G) \geq \sum_{u \in V(G)} w_u$, where 
\[
w_u = 
\begin{cases}
P[u \in S | u \in N[S]] & P[u \in N[S]] > 0\\
0 & \text{otherwise}.
\end{cases}
\]
\end{proposition}
\begin{IEEEproof}
For each $i \in \N$, let $T_i$ be an i.i.d. copy of $S$.
Define sequences $A_i, B_i$ as follows.
Initialize $A_0 = B_0 = \varnothing$.
Let $B_{i+1} = B_i \cup N[T_i]$ and let $A_{i+1} = T_i \setminus B_i$.
Our final independent set is $A_{\infty} = \bigcup_i A_i$.
Note that $B_i = \bigcup_{j=0}^i N[T_j] \supseteq \bigcup_{j=0}^i N[A_j]$.
Thus at step $i$ we exclude the vertices of $T_i$ that are adjacent to any members of $A_{j}$ for any $j < i$.

We have $E[|A_{\infty}|] = \sum_{u \in V(G)} P[u \in A_{\infty}]$.
If $P[u \in S] = 0$, then $P[u \in A_{\infty}] = 0$ as well.
For a vertex $u$ such that $P[u \in S] > 0$, $P[u \in B_{\infty}] = 1$, where $B_{\infty} = \bigcup_i B_i$.
We have
\[
\PR{u \in A_{i+1}} = \PR{u \in T_i, u \not\in B_i} = \PR{u \in T_i} \PR{u \not\in B_i}
\] 
because $B_i$ only depends on $T_j$ for $j < i$.
Now we have
\begin{IEEEeqnarray*}{rCl}
\PR{u \in A_{\infty}} 
&=& \sum_{i=0}^\infty \PR{u \in T_i | u \in N[T_i]} \PR{u \in N[T_i]} \PR{u \not\in B_i}\\
&=& \PR{u \in S | u \in N[S]} \sum_{i=0}^\infty \PR{u \in N[T_i], u \not\in B_i}\\
&=& \PR{u \in S | u \in N[S]} \PR{u \in B_{\infty}}.
\end{IEEEeqnarray*}
Thus $P[u \in A_{\infty}] = w_u$ for all $u$.
There must be some independent set in $G$ of size at least $E[|A_{\infty}|] = \sum_{u \in V(G)} w_u$.
\end{IEEEproof}

This is also a generalization of the Caro-Wei theorem \cite{alon_turans_2004}.
\begin{corollary}[Caro-Wei]
For any graph $G$, $\alpha(G) \geq \sum_{v \in V(G)} \frac{1}{d(v)+1}$.
\end{corollary}
\begin{proof}
Apply Proposition~\ref{prop:general-lb} with $S$ uniformly distributed over the singleton vertex sets.
\end{proof}
The following corollary will suffice for the rest of this paper.
\begin{corollary}
\label{cor:vt}
Let $G$ be a vertex-transitive graph and let $T$ be an independent set in $G$.
Then $\alpha(G) \geq \frac{|V(G)|\, |T|}{|N[T]|}$.
\end{corollary}
\begin{IEEEproof}
Apply Proposition~\ref{prop:general-lb} with $S$ as a translation of $T$ by an automorphism of $G$ chosen uniformly at random.
\end{IEEEproof}

\begin{theorem}
\label{thm:vt-ach}
Let $G$ be a vertex-transitive graph.
Then
\[
R^*(G,\delta) \geq \log \alpha(G) + R_{GV}\left(\frac{|V(G)|}{\alpha(G)},\delta\right).
\]
\end{theorem}
\begin{IEEEproof}
Let $S$ be a maximum independent set in $G$.
Then $S^n$ is independent in $G(n,d)$ and $|N[S^n]| = \sum_{i=0}^d \binom{n}{i}(|V(G)| - |S|)^i |S|^{n-i}$.
$G(n,d)$ is vertex-transitive, so by Corollary~\ref{cor:vt}
\[
\alpha(G(n,d)) \geq \frac{|V(G)|^n}{\sum_{i=0}^d \binom{n}{i}\left(\frac{|V(G)|}{|S|} - 1\right)^i }\,.
\]
Rewriting this inequality in terms of rates gives the claim.
\end{IEEEproof}

\section{Converse bound}
\label{sec:converse}
\begin{theorem} 
Let $G$ be vertex-transitive, be edge-transitive, and have at least one edge. 
Then 
\begin{equation}
R^*(G, \delta) \leq \log \theta_L(G) + R_{LP1}\left(\frac{|V(G)|}{\theta_L(G)},\, \delta\right)\,,
\end{equation}
where $R_{LP1}(q,\delta)$ was defined in~\eqref{eq:rlp1}.
\end{theorem}
\begin{remark}
For every edge-transitive $G$ we have~\cite{lovasz-1979}
\begin{equation}\label{eq:dtt0}
    \theta_L(G)=\frac{|V(G)|}{1-\frac{\lambda_0}{\lambda_m}}\,,
\end{equation}where $\lambda_0$ and $\lambda_m$ are the largest and
the smallest eigenvalues of the adiacency matrix of $G$, respectively.
\end{remark}
\begin{IEEEproof} 
Let $g=|V(G)|$. 
To bound $\alpha(G(n,\delta n))$ we use Schrijver-Delsarte's method \cite{schrijver-1979}.
For real symmetric matrices $T$ and $T'$, write $T \pgeq T'$ when $T - T'$ is positive semidefinite.
Let $\one$ be the column vector of ones.
For any graph $\Gamma$ we have
\[ \alpha(\Gamma) \le \theta_S(\Gamma) \]
where Schrijver's $\theta$-function $\theta_S(\Gamma)$ is defined as
\begin{equation}\label{eq:dtt12}
 \min\{\max_{v} T_{v,v}: T\pgeq \one\one^T, T_{v,v'} \le 0 \quad \forall (v,v') \in E(\overline \Gamma)\}\,. 
\end{equation}
Note that if the condition $T_{v,v'} \le 0$ is replaced with $T_{v,v'} = 0$ we get an alternative definition of the Lovasz'
$\theta_L(\Gamma)$. Denote by $D$ the $g\times g$ matrix achieving Lovasz's $\theta_L(G)$. 
For edge-transitive case it is known that 
\begin{equation}\label{eq:dtt}
  D = \frac{g}{ \lambda_0 - \lambda_m} (A_G - \lambda_m I)\,, \qquad \theta_L = \frac{-\lambda_m g}{\lambda_0-\lambda_m}
\end{equation}
where we enumerated $\lambda_0 \ge \cdots \ge \lambda_m$ the eigenvalues of adjacency matrix $A_G$. 
Note that $\tr A = 0$ but $A$ has some nonzero entries, so $\lambda_0 > 0$, $\lambda_m < 0$, and $D$ is entrywise non-negative.
Let $P_0$ and $P_m$ be orthogonal projectors on the space of constant functions and $\lambda_m$-eigenspace respectively. 
Thus $P_0 = \frac{1}{g}\one\one^T$. 

We will bound $\theta_S(G(n,\delta n))$ by optimizing over the restricted set of $T$'s in~\eqref{eq:dtt12}. Namely, 
for any $z \in \{0,m\}^n$ define 
\begin{align} P_z &\eqdef \bigotimes_{i=1}^n P_{z_i}\\
\Pi_\ell &\eqdef \sum_{z \in \{0,m\}^n: \|z\|_0 = \ell} P_z\,, \quad \ell=0,\ldots,n
\end{align}
where $\|\cdot\|_0$ is the Hamming weight. We will search $T$-assignments in the form 
\[ T = D^{\otimes n} \odot \left( \sum_{\ell=0}^n \hat h_\ell \Pi_\ell \right)\]
with $\hat h_{\ell}\ge 0$ and with $\odot$ denoting the Hadamard (entry-wise) product. 
We have to express the two conditions on $T$ from~\eqref{eq:dtt12} in terms of the coefficients $\hat h_\ell$. 

First, we consider the condition $T \pgeq \one\one^T$.
\begin{enumerate}
\item Since $D\odot P_0 = \frac{1}{g}D$, we have $(D \odot P_0) \one = \one$.
\item Note that 
\begin{equation}\label{eq:pmd}
    \img P_m = \ker D
\end{equation}
implying that $\tr (P_m D) =0$ and thus $\one^T(D\odot P_m) \one = 0\,$.    
So $\one$ is in the kernel of $D\odot P_m$. 
Similarly, $\one$ is in the kernel of $D^{\otimes n}\odot P_z$ for $\|z\|_0 > 0$.
\item Consequently, $\one$ is an eigenvector of $D^{\otimes n}\odot P_z$ for any $z$ and $\one$ is an eigenvector of $T$ for any choice of $\{\hat h_\ell\}$ and the eigenvalue of $\one$ is $\hat{h}_0$.
\item
Since Hadamard-product preserves positive-semidefiniteness, it is clear that $T \pgeq 0$.
Because $\one$ is an eigenvector of $T$, the condition $T \pgeq \one\one^T$ is equivalent to 
\begin{equation}\label{eq:dtt8}
g \hat{h}_0 = \one^T T \one \ge \one^T(\one\one^T)\one = g^{2n} 
\end{equation}
\end{enumerate}

Next, consider the condition $T_{v,v'} \le 0 \quad \forall (v,v') \in E(\overline{G(n,\delta n)})$, i.e. all $(v,v')$ such that $d(v,v') > \delta n$.
We have 
\begin{align*}
(D^{\otimes n})_{v,v'} =   0, & \quad d(v,v')=\infty \\
(D^{\otimes n})_{v,v'} \ge 0, & \quad d(v,v') < \infty\,.
\end{align*}
Thus we need
\[
\left( \sum_{\ell=0}^n \hat h_\ell \Pi_\ell \right)_{v,v'} \leq 0
\]
for all $(v,v')$ such that $\delta n < d(v,v') < \infty$.

Denote $d \eqdef \tr P_m$ (the dimension of $\lambda_m$-eigenspace) and $c \eqdef -g (P_m)_{v,v'}$ for any pair of adjacent vertices $(v,v')$. 
Note that, by edge-transitivity, $c$ does not depend on the choice of pair of vertices%
\footnote{This is the key reason for requiring edge-transitivity. For non-edge-transitive graphs, e.g. the complement of the Kneser graph $\overline K_{7,3}$ or $(C_5)^2$, we actually do not have the constancy of $c$ on the edge-set.}.
We can relate $c/d$ to spectrum of $A_G$ by using $\tr(P_m D)=0$:
\begin{equation}\label{eq:dtt7}
  d \lambda_m = \tr P_m^* A_G = - c \frac{|E(G)|}{g} = -c\lambda_0\,. 
\end{equation}    
In particular, $c>0$.

We now let $d(v,v')=d_0 < \infty$ and notice that this implies that for every $i \in [n]$, either $\{v_i, v_i'\} \in E(G)$
or $v_i=v_i'$. Therefore, under restriction of finite distance we have
\[ (P_m)_{v_i,v_i'} = \frac{1}{g} \left(d 1\{v_i=v_i'\} - c 1\{v_i \neq v_i'\}\right)\,.\]
Consequently, 
\[ (P_z)_{v,v'} = \frac{1}{g^n} (-c)^{b} d^{\|z\|_0 - b}\,, \,\, b=|\{i: v_i\neq v_i', z_i=m\}|\,.\]
Finally, summing over all $z$ with Hamming weight $\ell$ we get 
\[ (\Pi_\ell)_{v,v'} = \frac{1}{g^n}c^{\ell} K_\ell(d(v,v'))\,, \]
where we introduce Krawtchouk polynomials
\[ K_\ell(x) \eqdef \sum_{j=0}^\ell \binom{x}{j} \binom{n-x}{\ell - j} (-1)^j (q'-1)^{\ell-j}\,, \]
and $q' = 1+\frac{d}{c} = 1- \frac{\lambda_0}{\lambda_m} = \frac{g}{\theta_L(G)}$ by~\eqref{eq:dtt0}
and~\eqref{eq:dtt7}.

Thus $T_{v,v'} \le 0 $ for $\delta n < d(v,v')$ is equivalent to
\begin{equation}\label{eq:dtt9}
  H(x) \le 0 \qquad \forall x \in \mathbb{Z} \cap [\delta n, n] 
\end{equation}
where we introduce
\[ H(x) = \sum_{\ell=0}^n \hat H_\ell K_\ell(x)\,, \quad \hat H_\ell \eqdef \frac{1}{g^n} c^\ell \hat h_\ell\,. \]

Relaxing the constraint in~\eqref{eq:dtt9} to $H(x) \le 0$ on the interval $[\delta n,n]$ we get the problem:
\begin{equation}\label{eq:htoLPq'}
    A_{LP1}(n,\delta n) \eqdef \min\left\{\frac{H(0)}{\hat H_0}: \hat H_\ell \ge 0, H(x) \le 0 \quad
\forall x\in[\delta n, n]\right\}\,. 
\end{equation}
Since $D_{v,v} = \theta_L(G)$ the overall bound becomes:
\begin{equation}    \alpha(G(n,\delta n)) \le \theta_L(G)^n  A_{LP1}(n,\delta n)\,.
\end{equation}
The minimization of \eqref{eq:htoLPq'} is what is known as the first linear programming problem for the Hamming space,
though with a non-integer parameter $q'$.
Although exact asymptotics of~\eqref{eq:htoLPq'} is hereto unknown even in the binary case,
cf.~\cite{navon2005delsarte}, we can use the standard MRRW choice of the polynomial $H(x) = \frac{1}{x-a}(K_t(a)
K_{t+1}(x)-K_{t+1}(a)K_t(x))^2$, see~\cite{mceliece-et-al-1977} for the choice of $a$ and $t$. Their arguments 
can be applied verbatim for non-integer values of $q'$ (see also \cite{ismail-simeonov-1998}
for the position of the roots of $K_\ell(x)$) and it implies
\[ A_{LP1}(n, \delta n) \le \exp\{n R_{LP1}(q', \delta) + o(n)\}\,,\]
and the claim of the theorem follows.
\end{IEEEproof}

\section{Relations between graphs}
\label{sec:graph-stuff}

In this section we summarize some of the methods that can be useful for extending the previous basic results to other
graphs (possibly lacking symmetries).

\subsection{Bounds from graph homomorphisms}
A function $f : V(G) \to V(H)$ is graph homomorphism from $G$ to $H$ if $\{u, v\} \in E(G)$ implies $\{f(u), f(v)\} \in E(H)$.
We will write $f : G \to H$ to indicate that $f$ is a homomorphism, or just $G \to H$ to indicate that a homomorphism exists.

\begin{proposition}
\label{prop:lift}
If there is some $f : G \to H$, then $f^{\otimes n} : G(n,d) \to H(n,d)$. 
If additionally $H$ is vertex-transitive then
\[ R^*(G, \delta) \ge \log \frac{|V(G)|}{|V(H)|} + R^*(H,\delta)\,.\]
\end{proposition}
\begin{IEEEproof}
For $u,v \in V(G)$, if $d_G(u,v) < \infty$ then $d_H(f(u), f(v)) = d_G(u,v)$.
This property extends to the semimetrics on $V(G)^n$ and $V(H)^n$.

If $H$ is vertex-transitive, $H(n,d)$ is as well.
Because $G(n,d) \to H(n,d)$ and $H(n,d)$ is vertex-transitive, we have
\[\frac{\alpha(G(n,d))}{|V(G)|^n} \geq \frac{\alpha(H(n,d))}{|V(H)|^n}\] by the No-Homomorphism Lemma~\cite{AC85}.
Rewriting this inequality in terms of rates gives the claim.
\end{IEEEproof}

In particular, for any $c$-colorable graph $G$ we have, by applying the previous proposition to $G\to K_c$, that
\begin{equation}\label{eq:ach1}
    R^*(G, \delta) \ge \log \frac{|V(G)|}{c} + R^*(K_c, \delta)\,.
\end{equation}
We may further lower-bound $R^*(K_c,\delta)$ by the GV-bound $R_{GV}(c, \delta)$.
This may not be the best one known, though!
It turns out that this latter bound can be improved, as next section shows, by replacing coloring with fractional coloring.  

\subsection{Bounds from fractional chromatic number}

For vertex-transitive graphs we have $\alpha(G)\chi^*(G) = |V(G)|$ and thus we may restate Theorem~\ref{thm:vt-ach} as
\begin{equation}\label{eq:ach2}
R^*(G,\delta) \ge \log \frac{|V(G)|}{\chi^*(G)} + R_{GV}(\chi^*(G), \delta)\,.
\end{equation}
It turns out~\eqref{eq:ach2} holds even without vertex-transitivity.
\begin{proposition}
For any graph $G$ the bound~\eqref{eq:ach2} holds.
\end{proposition}
\begin{IEEEproof}
If $\chi^*(G) = p/q$, then for some positive integer $b$, $G \to K_{bp,bq}$, where $K_{bp,bq}$ is a Kneser graph as defined in Section~\ref{sec:examples} 
\cite{godsil_algebraic_2013}.
$K_{bp,bq}$ is vertex-transitive and $\chi^*(K_{bp,bq}) = \frac{p}{q}$.
Thus from Proposition~\ref{prop:lift} and \eqref{eq:ach2},
\begin{IEEEeqnarray*}{rCl}
R^*(G,\delta)
&\geq& \log \frac{|V(G)|}{|K_{bp,bq}|} + R^*(K_{bp,bq},\delta)\\
&\geq& \log \frac{|V(G)|}{|K_{bp,bq}|} + \log \frac{|K_{bp,bq}|}{p/q} + R_{GV}\left(\frac{p}{q},\delta\right)\,.\, 
\end{IEEEeqnarray*}
\end{IEEEproof}

\subsection{Bounds from powers of graphs}
\begin{proposition}
\label{prop:using-powers}
\begin{equation*}
\frac{1}{r} R^*(G^r, r\delta) \leq R^*(G, \delta) \leq \frac{1}{r} R^*(G^r, \delta)
\end{equation*}
\end{proposition}
\begin{IEEEproof}
These three graphs can all have vertex set $V(G)^{rn}$. 
We have $E(G^r(n,d)) \subseteq E(G(rn,rd)) \subseteq E(G^r(n,rd))$ and 
$\alpha(G^r(n,rd)) \leq \alpha(G(rn,rd)) \leq \alpha(G^r(n,d))$.
\end{IEEEproof}

A useful application of Proposition \ref{prop:using-powers} is the lower bound we presented for $C_5$ in equation
\eqref{eq:LB-C_5^2}. In this particular case, we exploit the fact that $C_5^2$ is $5$-colorable to use equation
\eqref{eq:ach1} on $C_5^2$ and then Proposition \ref{prop:using-powers} to bound $R^*(C_5,\delta)$.

\subsection{Clique covers}
In some special cases, we can show that distinct graphs have essentially identical rate functions.
\begin{proposition}
\label{prop:tight-cover}
Let $G$ be a graph with $\omega(G) = \chi(G) = c$, and $\theta^*(G)\omega(G) = |V(G)|$. 
Then $\alpha(G) = \theta^*(G)$ and 
\begin{equation*}
R^*(G,\delta) = \log \alpha(G) + R^*(K_c, \delta).
\end{equation*}
\end{proposition}
\begin{IEEEproof}
For any graph $H$, $\frac{|V(H)|}{\chi(H)} \leq \alpha(H)$ because some color class is at least as large as the average and $\alpha(G) \leq \theta^*(H)$ by the pigeonhole principle.
If $H$ satisfies the conditions of this proposition, the inequalities must be tight.

Let $\mats$ be the set of maximal cliques in $G$.
There is a fractional clique covering of $G$ of weight $\theta^*(G)$.
This means that there is some weight vector $w \in \R_+^{\mats}$ such that $\one^Tw = \theta^*(G)$ and $\sum_{S:u \in S} w_S \geq 1$ for all $u \in V(G)$.
Because $\theta^*(G)\omega(G) = |V(G)|$, only cliques of size $\omega(G)$ are assigned nonzero weight.

For each $S \in \mats$, $G[S]$ (the subgraph of $G$ induced by $S$) is isomorphic to $K_c$.
For each $T =T_1 \times \ldots \times T_{|\mats|} \in \mats^n$, $G(n,d)[T] = K_c(n,d)$.
Thus
\begin{IEEEeqnarray}{rCl}
\alpha(G(n,d)) 
&\leq& \sum_{T \in \mats^n} \alpha(G(n,d)[T]) \prod_{i \in [n]} w_{T_i}\nonumber\\
&\leq& \alpha(K_c(n,d)) \sum_{T \in \mats^n} \prod_{i \in [n]} w_{T_i}\nonumber\\
&=& \alpha(K_c(n,d)) \left( \sum_{S \in \mats} w_S\right)^n\nonumber\\
&=& \theta^*(G)^n \alpha(K_c(n,d))\label{cover}
\end{IEEEeqnarray}
There is a homomorphism $G \to K_{\chi(G)}$, so from \eqref{eq:ach1} and \eqref{cover}, 
\begin{equation*}
\log \frac{|V(G)|}{\chi^*(G)} + R^*(K_c, \delta)
\leq R^*(G,\delta)
\leq \log \theta^*(G) + R^*(K_c, \delta)\,.
\end{equation*}
These bounds are both equal to $\log \alpha(G) + R^*(K_c, \delta)$.
\end{IEEEproof}

\subsection{Sum of cliques}
The simplicity of Proposition~\ref{prop:tight-cover} depends on the fact that all cliques involved are the same size.
When every optimal clique cover involves multiple sizes of cliques, complications ensue.
\begin{proposition}
\label{sec:unequal-cliques}
Let $G= a_1 K_1 + a_c K_c$ and let $q = \frac{a_c}{a_i}+1$.
Then
\[ 
R^*(G,\delta) = \log a_1 + \max_{0 \leq \lambda \leq 1} \left[ H_q(\lambda) + \lambda R^*\left(K_c,\frac{\delta}{\lambda}\right)\right].
\]
\end{proposition}
\begin{IEEEproof}
$G(n,d) = \sum_i \binom{n}{i}a_c^ia_1^{n-i} K_c(i,d)$.
$R^*(G,\delta)$ depends only on the largest term in this sum.
Let $i = \lambda n$.
We have $\limsup_{n\to\infty} \lambda \frac{1}{\lambda n} \log K_c(\lambda n,\delta n) = p R^*\left(K_c,\frac{\delta}{\lambda}\right)$ and $\binom{n}{i}a_c^ia_1^{n-i} = a_1^n\binom{n}{i}(q-1)^i$.
\end{IEEEproof}

\bibliographystyle{IEEEtran}

\begin{thebibliography}{10}
\providecommand{\url}[1]{#1}
\csname url@samestyle\endcsname
\providecommand{\newblock}{\relax}
\providecommand{\bibinfo}[2]{#2}
\providecommand{\BIBentrySTDinterwordspacing}{\spaceskip=0pt\relax}
\providecommand{\BIBentryALTinterwordstretchfactor}{4}
\providecommand{\BIBentryALTinterwordspacing}{\spaceskip=\fontdimen2\font plus
\BIBentryALTinterwordstretchfactor\fontdimen3\font minus
  \fontdimen4\font\relax}
\providecommand{\BIBforeignlanguage}[2]{{%
\expandafter\ifx\csname l@#1\endcsname\relax
\typeout{** WARNING: IEEEtran.bst: No hyphenation pattern has been}%
\typeout{** loaded for the language `#1'. Using the pattern for}%
\typeout{** the default language instead.}%
\else
\language=\csname l@#1\endcsname
\fi
#2}}
\providecommand{\BIBdecl}{\relax}
\BIBdecl

\bibitem{dalai-TIT-2015}
M.~Dalai, ``Elias {B}ound for {G}eneral {D}istances and {S}table {S}ets in
  {E}dge-{W}eighted {G}raphs,'' \emph{IEEE Trans. Inform. Theory}, vol.~61,
  no.~5, pp. 2335--2350, May 2015.

\bibitem{shannon-1956}
C.~E. Shannon, ``{T}he {Z}ero-{E}rror {C}apacity of a {N}oisy {C}hannel,''
  \emph{IRE Trans. Inform. Theory}, vol. IT-2, pp. 8--19, 1956.

\bibitem{lovasz-1979}
L.~Lov{\'a}sz, ``{O}n the {S}hannon {C}apacity of a {G}raph,'' \emph{IEEE
  Trans. Inform. Theory}, vol.~25, no.~1, pp. 1--7, 1979.

\bibitem{aaltonen1990new}
M.~Aaltonen, ``A new upper bound on nonbinary block codes,'' \emph{Discrete
  Mathematics}, vol.~83, no.~2, pp. 139--160, 1990.

\bibitem{tsfasman1982modular}
M.~A. Tsfasman, S.~Vl{\u{a}}dut, and T.~Zink, ``Modular curves, {S}himura
  curves, and {G}oppa codes, better than {V}arshamov-{G}ilbert bound,''
  \emph{Mathematische Nachrichten}, vol. 109, no.~1, pp. 21--28, 1982.

\bibitem{west_introduction_2001}
D.~B. West, \emph{Introduction to graph theory}.\hskip 1em plus 0.5em minus
  0.4em\relax Prentice Hall Upper Saddle River, 2001, vol.~2.

\bibitem{alon_turans_2004}
N.~Alon and J.~H. Spencer, ``Tur\'{a}n's theorem,'' in \emph{The probabilistic
  method}.\hskip 1em plus 0.5em minus 0.4em\relax John Wiley \& Sons, 2004, pp.
  95--96.

\bibitem{schrijver-1979}
A.~Schrijver, ``A comparison of the {D}elsarte and {L}ov{\'a}sz bounds,''
  \emph{IEEE Trans. on Inform. Theory}, vol.~25, no.~4, pp. 425 -- 429, jul
  1979.

\bibitem{navon2005delsarte}
M.~Navon and A.~Samorodnitsky, ``On {D}elsarte's linear programming bounds for
  binary codes,'' in \emph{Proc. 46th Annual IEEE Symp.Found. Comp. Sci.
  (FOCS)}.\hskip 1em plus 0.5em minus 0.4em\relax IEEE Computer Society, 2005,
  pp. 327--338.

\bibitem{mceliece-et-al-1977}
R.~McEliece, E.~Rodemich, H.~Rumsey, and L.~Welch, ``New upper bounds on the
  rate of a code via the {D}elsarte-{M}ac{W}illiams inequalities,''
  \emph{Information Theory, IEEE Transactions on}, vol.~23, no.~2, pp. 157 --
  166, mar 1977.

\bibitem{ismail-simeonov-1998}
M.~E.~H. Ismail and P.~Simeonov, ``Strong {A}symptotics for {K}rawtchouk
  {P}olynomials,'' \emph{J. Comput. Appl. Math.}, vol. 100, no.~2, pp.
  121--144, 1998.

\bibitem{AC85}
M.~O. Albertson and K.~L. Collins, ``Homomorphisms of 3-chromatic graphs,''
  \emph{Discr. Math.}, vol.~54, no.~2, pp. 127--132, 1985.

\bibitem{godsil_algebraic_2013}
C.~Godsil and G.~F. Royle, \emph{Fractional Chromatic Number}.\hskip 1em plus
  0.5em minus 0.4em\relax Springer Science \& Business Media, 2013, vol. 207.

\end{thebibliography}

\end{document}